\documentclass[11pt]{article}
\usepackage[utf8]{inputenc}

\usepackage[margin=1in]{geometry}
\usepackage{amsmath}
\usepackage{amssymb}
\usepackage{amsthm}
\usepackage{algorithm}
\usepackage{algpseudocode}
\usepackage{hyperref}
\usepackage{graphicx}
\usepackage{subcaption}
\usepackage[font={small}]{caption}
\usepackage{thm-restate}
\usepackage[normalem]{ulem}
\usepackage{framed}
\usepackage{enumitem}

\usepackage[dvipsnames]{xcolor}

\hypersetup{
pagebackref,
colorlinks=true,
urlcolor=blue,
linkcolor=blue,
citecolor=OliveGreen,
}

\newcommand{\nc}{\newcommand}
\newcommand{\DMO}{\DeclareMathOperator}

\algnotext{EndFor}
\algnotext{EndIf}
\algnotext{EndWhile}
\algnotext{EndProcedure}

\allowdisplaybreaks

\nc{\MS}{\mathcal{S}}
\nc{\MR}{\mathcal{R}}
\nc{\cM}{\mathcal{M}}

\nc{\MZ}{\mathcal{Z}}

\DMO{\Binom}{Binom}
\newcommand{\E}{\mathbb{E}}
\DMO{\Var}{Var}

\newcommand{\N}{\mathbb{N}}

\nc{\BN}{\mathbb{N}}

\nc{\BZ}{\mathbb{Z}}

\newcommand{\eps}{\varepsilon}
\nc{\ep}{\eps}

\newcommand{\cA}{\mathcal{A}}

\DeclareMathOperator{\svt}{IterativeSVT}
\DeclareMathOperator{\abt}{AboveThreshold}
\DeclareMathOperator{\pabt}{PermutedAboveThreshold}

\DeclareMathOperator{\Lap}{Lap}

\DeclareMathOperator{\poly}{poly}
\renewcommand{\varepsilon}{\epsilon}
\let\Pr=\relax
\DeclareMathOperator*{\Pr}{\mathbf{Pr}}

\newcommand{\bound}[3]{\mathsf{err}_{#1, #2, #3}}
\newcommand{\sbound}{\bound{k}{\epsilon}{\delta}}
\newcommand{\SV}{\mathtt{sv}}
\newcommand{\good}{\mathrm{good}}
\newcommand{\bad}{\mathrm{bad}}

\nc{\rgp}{\mathrm{RP}}

\newtheorem{theorem}{Theorem}
\newtheorem*{theorem*}{Theorem}

\newtheorem{observation}[theorem]{Observation}
\newtheorem{lemma}[theorem]{Lemma}
\newtheorem*{lemma*}{Lemma}
\newtheorem*{obs*}{Observation}

\newtheorem{definition}[theorem]{Definition}

\newcount\Comments
\newcommand{\badih}[1]{\ifnum\Comments=1\textcolor{red}{[Badih: #1]}\fi}
\newcommand{\pasin}[1]{\ifnum\Comments=1\textcolor{red}{[Pasin: #1]}\fi}
\newcommand{\ravi}[1]{\ifnum\Comments=1\textcolor{cyan}{[Ravi: #1]}\fi}

\title{On Avoiding the Union Bound When Answering \\ Multiple Differentially Private Queries}
\date{\today}

\author{
  Badih Ghazi\thanks{Google Research, Mountain View, CA. Email: \texttt{badihghazi@gmail.com}}
  \and
  Ravi Kumar\thanks{Google Research, Mountain View, CA. Email: \texttt{ravi.k53@gmail.com}}
  \and
  Pasin Manurangsi\thanks{Google Research, Mountain View, CA. Email: \texttt{pasin@google.com}}
}

\begin{document}

\maketitle

\begin{abstract}
In this work, we study the problem of answering $k$ queries with $(\eps, \delta)$-differential privacy, where each query has sensitivity one. We give an algorithm for this task that achieves an expected $\ell_\infty$ error bound of $O(\frac{1}{\epsilon}\sqrt{k \log \frac{1}{\delta}})$, which is known to be tight~\cite{SteinkeU16}. 

A very recent work by Dagan and Kur~\cite{DK20} provides a similar result, albeit via a completely different approach. One difference between our work and theirs is that our guarantee holds even when $\delta < 2^{-\Omega(k/(\log k)^8)}$ whereas theirs does not apply in this case.
On the other hand, the algorithm of~\cite{DK20} has a remarkable advantage that the $\ell_{\infty}$ error bound of $O(\frac{1}{\epsilon}\sqrt{k \log \frac{1}{\delta}})$ holds not only in expectation but always (i.e., with probability one) while we can only get a high probability (or expected) guarantee on the error.

\end{abstract}

\section{Introduction}

One of the most ubiquitous---as well as one of the first---differentially private (DP) algorithm is the Laplace mechanism~\cite{DworkMNS06} where, to answer some query $q$ on a sensitive dataset $X$, we simply compute the true answer $q(X)$ and then add to it a noise term sampled from the Laplace distribution, where the parameter of the distribution is calibrated to the desired privacy level and the sensitivity of $q$. For $\eps$-DP and when the query $q$ has sensitivity at most one, this algorithm yields an expected error of $O(\frac{1}{\eps})$, which is known to be tight \cite{ghosh2012universally}.

In real-world applications, however, it is rarely the case that only a single query is performed on the dataset $X$. A more realistic scenario is when we are given multiple queries $q_1, \dots, q_k$ to the dataset and are asked to compute private answers $a_1, \dots, a_k$ to these queries. While there are several measures of error that can be used, one of the most common is the $\ell_\infty$ error (aka maximum error), which is defined as $\max_{i \in [k]} |q_i(X) - a_i|$.

When the Laplace mechanism is applied in this multiple query setting, the privacy budget has to be split over the $k$ queries, i.e., each query has a budget of $\frac{\eps}{k}$. However, this does \emph{not} result in an $\ell_\infty$ error of $O(\frac{k}{\eps})$ because one has to apply a union bound over all the $k$ queries, which instead results in the expected $\ell_\infty$ error of $O(\frac{k \log k}{\epsilon})$. Remarkably, Steinke and Ullman~\cite{SteinkeU16} showed that this bound is \emph{not} tight, by giving an algorithm with expected $\ell_\infty$ error of $O(\frac{k}{\eps})$. In other words, their algorithm ``avoids the union bound'' in the error. Furthermore, this error is known to be asymptotically tight for $\eps$-DP~\cite{HardtT10}.


For $(\eps, \delta)$-DP algorithms \cite{DworkKMMN06}, the situation is more complicated.  For simplicity, throughout the paper, let 
\[
\sbound := \frac{1}{\epsilon} \sqrt{k \log \frac{1}{\delta}}.
\]
It is known that the expected $\ell_\infty$ error must be at least $\Omega(\sbound)$ for any\footnote{Note that the lower bound on $\delta$ is necessary, as the $\eps$-DP algorithm mentioned in the previous paragraph already yields an $O(\frac{k}{\eps})$ expected $\ell_\infty$ error.} $k^{-O(1)} \geq \delta \geq 2^{-\Omega(k/\eps)}$~\cite{SteinkeU16}. However, the Laplace mechanism, together with the advanced composition theorem~\cite{DworkRV10}, only gives a bound of $O(\sbound \cdot \log k)$, where the $\log k$ factor once again comes from applying the union bound over all $k$ coordinates. The Gaussian mechanism (see e.g.,~\cite{DworkR14}) gives an improved bound of $O(\sbound \cdot \sqrt{\log k})$ due to a better tail behavior of the noise distribution. Steinke and Ullman~\cite{SteinkeU16} once again showed that this is not optimal, by giving an algorithm with expected $\ell_\infty$ error of only $O(\sbound \cdot 
\sqrt{\log \log k})$. This has recently been improved by Ganesh and Zhao~\cite{GZ20} to $O(\sbound \cdot 
\sqrt{\log \log \log k}))$. Even more recently, Dagan and Kur~\cite{DK20} reduce this expected $\ell_\infty$ error to the optimal $O(\sbound)$ although their algorithm only works when $\delta$ is at least $2^{-\Omega(k/(\log k)^8)}$, thereby leaving open the question in the case $2^{-\Omega(k)} \leq \delta \leq 2^{-\Omega(k/(\log k)^8)}$.

\subsection{Our Contributions}

In this work, we resolve the question of~\cite{SteinkeU16} by presenting an $(\eps, \delta)$-DP algorithm with tight expected $\ell_\infty$ error for any $0.5 \geq \delta > 0$, including the regime $2^{-\Omega(k)} \leq \delta \leq 2^{-\Omega(k/(\log k)^8)}$ not covered by~\cite{DK20}.  Our main theorem is the following.

\begin{theorem} \label{thm:main}
For any $k \in \N, \eps \in (0, 1]$ and $\delta \in (0, 0.5]$, there exists an $(\eps, \delta)$-DP algorithm that can answer $k$ queries, each of sensitivity at most one, with expected $\ell_\infty$ error $O(\sbound)$.
\end{theorem}

\paragraph{Differences from~\cite{DK20}.} 
We stress that the techniques used in our work and~\cite{DK20} are completely different. Specifically, Dagan and Kur~\cite{DK20} arrived at their result by designing a new noise distribution and analyze the algorithm that adds such independent noise to each query's answer. On the other hand, our approach, which is detailed in the next section, is based on the \emph{sparse vector technique}~\cite{DworkNRRV09,HardtR10,RothR10,DworkNPR10} similar to that of~\cite{SteinkeU16}.

In terms of the guarantees, we reiterate that our algorithm works for any $\delta \leq 0.5$, whereas the current analysis of the algorithm of~\cite{DK20} does not apply for $\delta \leq 2^{-\Omega(k/(\log k)^8)}$. On the other hand, the algorithm of~\cite{DK20} has a remarkable advantage that the $\ell_{\infty}$ error bound of $O(\sbound)$ holds not only in expectation but always (i.e., with probability one). In contrast, we can only get a high probability guarantee that the $\ell_{\infty}$ error does not exceed this bound (see Theorem~\ref{thm:high-prob-error}).

\subsection{Proof Overview}
In this section, we describe the high-level technical ideas of our algorithm. We will sometimes be informal here, but all the details will be formalized in subsequent sections.

Our algorithm is inspired by the work of Steinke and Ullman~\cite{SteinkeU16}. Their algorithm works by first adding Gaussian noise to the queries. Then, they use the so-called \emph{sparse vector technique}~\cite{DworkNRRV09,HardtR10,RothR10,DworkNPR10} to ``correct'' the answers that are too far away from the true answers. The procedure they employed in this correction step is encapsulated in the following theorem; its proof can be found, e.g., in~\cite{DworkR14}\footnote{See also~\cite[Theorem 18]{GZ20} for a more detailed explanation.}.

\begin{theorem} \label{thm:sv-correction}
For every $k \geq 1, c_{\SV} \leq k, \eps_{\SV}, \delta_{\SV}, \beta_{\SV} > 0$, and
\begin{align*}
\alpha_{\SV} \geq O\left(\bound{c_\SV}{\eps_\SV}{\delta_\SV}
\cdot \log \frac{k}{\beta_{\SV}} \right),
\end{align*}
there exists an $(\eps_{\SV}, \delta_{\SV})$-DP algorithm that takes as input queries $g_1, \dots, g_k$ each of sensitivity one and if there are at most $c_{\SV}$ indices $i \in [k]$ such that $|g_i(X)| > \alpha_{\SV}/2$, then, with probability at least $1 - \beta_{\SV}$, it answers all the queries with $\ell_\infty$ error is no more than $\alpha_{\SV}$.
\end{theorem}

Note here that $g_i$ should be thought of as the difference between the true answer $q_i$ and the estimate output in the first step. The above algorithm can be used to ``correct'' the $g_i$'s that are too large, if there are not too many of them. Specifically, notice that the error $\alpha_{\SV}$ is $O(\sbound)$ only when the number of ``very incorrect'' answers $c_{\SV}$ is at most $O(k / \log^2 k)$. This is indeed the reason why~\cite{SteinkeU16} achieve an error of $O(\sbound \cdot \sqrt{\log\log k})$, as they need to ensure (using the tail bound for Gaussian noise) that at most $O(k / \log^2 k)$ coordinates are ``very incorrect''.

This brings us to the main technical question we explore in this work: can we still apply the correction procedure when $\omega(k/\log^2 k)$ coordinates are ``very incorrect''? How about even at $\Omega(k)$? In other words, can we apply a sparse vector-based correction in the \emph{dense} regime?

In a specific sense, we show that this is possible, by carefully applying the sparse vector technique iteratively and ensuring that (with high probability) some progress is made each time.

To be more specific, we have to understand how the $\poly\log(k)$ factor appears in the first place. Roughly speaking, the main primitive used in Theorem~\ref{thm:sv-correction} is the following $\abt$ algorithm, which allows us to identify a single ``incorrect'' coordinate.

\begin{algorithm}[!htp]
\begin{algorithmic}[1]
  \Procedure{$\abt^\eps_T$}{$g_1, g_2,  \dots, g_k$; $X$}
  \State $\rho \leftarrow \Lap(2/\eps)$
  \For{$i = 1, \dots, k$}
  \State $\nu_i \leftarrow \Lap(4/\eps)$
  \If{$g_i(X)  + \nu_i \geq T + \rho$}
  \State \Return $i$
  \EndIf
  \EndFor
  \EndProcedure
 \end{algorithmic}
\caption{AboveThreshold Algorithm.}\label{alg:abt}
\end{algorithm}
Its privacy guarantee is well-known (see, e.g., \cite[Theorem 3.23]{DworkR14}):

\begin{theorem} \label{thm:abt-privacy}
If each of $g_1, g_2, \dots, g_k$ has sensitivity at most one, Algorithm~\ref{alg:abt} is $\eps$-DP.
\end{theorem} 

It turns out that the $\log k$ factor in $\alpha_{\SV}$ comes from a rather extreme situation: suppose that $g_1, \dots, g_{k/2}$ are the ``correct'' coordinates, e.g., $g_1 = \cdots = g_{k/2} = 0$. To make sure that we do not output these coordinates we have to  make sure that \emph{all of} $\nu_1, \dots, \nu_{k/2}$ are smaller than the threshold. Thus, the threshold has to be at least $\Omega(\log k)$. We end by noting that the $\sqrt{c_{\SV} \log \frac{1}{\delta_{\SV}}}$ factor in $\alpha_{\SV}$ then shows up because of $c_{\SV}$-fold advanced composition (see Theorem~\ref{thm:adv-composition}).

Now, the above example is extreme. In fact, suppose that there is a $\gamma$ fraction of the coordinates that we would like to correct. If we randomly permute the coordinates, at least one of these coordinates will appear, with a constant probability, in the first $1/\gamma$ coordinates. As a result, we only have to ensure that $\rho, \nu_1, \dots, \nu_{1/\gamma}$ are small, meaning that we should be able to get away with a threshold of $\log \frac{1}{\gamma}$ instead of $\log k$.

Our algorithm formalizes this idea. Specifically, it works in stages. In stage $\ell$, we have a target number $m_{\ell}$ of items we would like to ``correct''. This number will be (slowly) geometrically decreasing. The $\eps$'s for (the permuted version of) the $\abt$ algorithm in each stage on the other hand geometrically increase, but slower than $m_\ell$ so that the entire algorithm in the end remains $(\eps, \delta)$-DP.

The actual analysis is more involved than the above outline, because some of the ``correction'' operations can also flip a ``correct'' coordinate to an ``incorrect'' one, and as such we have to track this number as well in order to ensure that we make progress. Another technical point is that while we can analyze the algorithm until no ``incorrect'' coordinates remain at all, it turns out to be more complicated as we need finer concentration inequalities.  Instead, we analyze our iterative algorithm until the number of ``incorrect'' coordinates is sufficiently small that we can apply Theorem~\ref{thm:sv-correction}. Finally, even after doing so, it only gives a high probability bound over the $\ell_\infty$ error, so we then devise a simple extension that roughly takes the best of the current output and the output of another application of the Gaussian mechanism, which helps us bound the expected $\ell_\infty$ error, which eventually yields Theorem~\ref{thm:main}.

\paragraph{Organization.}
We provide necessary background in Section~\ref{sec:prelim}. Then, we start by analyzing the ``permuted'' variant of the $\abt$ algorithm in Section~\ref{sec:pabt}. We continue in Section~\ref{sec:isvc} by presenting our iterative correction algorithm and give an upper bound on the number of ``incorrect'' coordinates. In Section~\ref{sec:high-prob-error}, we use this together with Theorem~\ref{thm:sv-correction} to obtain a high-probability $\ell_\infty$ error bound. We then use this to obtain the expected error bound in Section~\ref{sec:expected-error}. Finally, we discuss some open questions in Section~\ref{sec:open}.

\section{Preliminaries}
\label{sec:prelim}

\begin{definition}[Differential Privacy~\cite{DworkMNS06, DworkKMMN06}]
For $\eps, \delta \geq 0$, we say that an algorithm $\cA$ is $(\eps, \delta)$-differentially private (or $(\eps, \delta)$-DP for short) if the following holds for any set $S$ of outputs and any neighboring datasets $X, X'$:
\begin{align*}
\Pr[\cA(X) \in S] \leq e^{\eps} \cdot \Pr[\cA(X') \in S] + \delta.
\end{align*}
When $\delta = 0$, we may simply say that the algorithm is $\eps$-DP.
\end{definition}

The \emph{sensitivity} of a real-valued function $f$ is defined to be $\max_{X, X'} |f(X) - f(X')|$ where the maximum is taken over all pairs of neighboring datasets $X, X'$. For the purpose of this work, it does not matter how the neighboring relationship is defined; as long as we measure the sensitivity and differential privacy under the same neighboring notion, the results hold.

For brevity, we will assume henceforth that any query considered has sensitivity at most one and we may not state this assumption explicitly. Note that all results extend to the case where the sensitivity is bounded by $\Delta$, with the (necessary) multiplicative $\Delta$ factor in the expected error.

We also recall the following composition theorems for DP:

\begin{theorem}[Basic Composition Theorem]
\label{thm:basic-composition}An algorithm that applies a sequence of $(\eps_1, \delta_1)$-DP, $\dots, (\eps_m, \delta_m)$-DP (possibly adaptive) algorithms is $(\eps_1 + \cdots + \eps_m, \delta_1 + \cdots + \delta_m)$-DP.
\end{theorem}

\begin{theorem}[Advanced Composition Theorem~\cite{DworkRV10}] \label{thm:adv-composition}
An algorithm that applies an $\eps$-DP (possibly adaptive) algorithm $m$ times is $(\eps', \delta')$-DP for any $\delta' > 0$ with
\begin{align*}
\eps' = \sqrt{2m\log \frac{1}{\delta'}} \cdot \eps + m\eps(e^{\eps} - 1).
\end{align*}
\end{theorem}

\section{Permuted AboveThreshold Algorithm}
\label{sec:pabt}

We start by analyzing the variant of the $\abt$ algorithm where the coordinates are randomly permuted, as presented in Algorithm~\ref{alg:pabt}. 

\begin{algorithm}[!htp]
\begin{algorithmic}[1]
  \Procedure{$\pabt^\eps_T$}{$f_1, f_2,  \dots, f_k$; $X$}
  \State $\pi \leftarrow$ random permutation on $[k]$
  \State $i \leftarrow \abt^\eps_T(f_{\pi(1)}, f_{\pi(2)}, \dots, f_{\pi(k)}; X)$ \label{step:abt-call}
  \State \Return $\pi^{-1}(i)$
  \EndProcedure
 \end{algorithmic}
\caption{Permuted AboveThreshold Algorithm.}\label{alg:pabt}
\end{algorithm}

Since the only step in Algorithm~\ref{alg:pabt} that depends on the input dataset is Step~\ref{step:abt-call}, we can apply the privacy guarantee of $\abt$ (from Theorem~\ref{thm:abt-privacy}) to arrive at the following similar guarantee for $\pabt$.

\begin{observation} \label{obs:pabt-privacy}
If each of $f_1, f_2, \ldots, f_k$ has sensitivity at most one, Algorithm~\ref{alg:pabt} is $\eps$-DP.
\end{observation}

Next, we analyze its utility. Let $i^*$ be the output index of $\pabt^\eps_T$. Our goal here is to ensure that, with some non-trivial constant probability, $f_{i^*}(X) \geq T + w$ for some parameter $w > 0$. Similar to the known analyses of the vanilla $\abt$ algorithm, we will have to assume that there are not too many ``bad'' coordinates $i$ with $f_i(X) \in (T - w, T + w)$. The difference in our analysis below is that we additionally assume that there are many, i.e., $\gamma \cdot k$ ``good'' coordinates $i$ that satisfy $f_i(X) \geq T + w$. This turns out to help us reduce the $w$ parameter; specifically, we get $w$ to be as small as $O(\frac{1}{\epsilon} \log \frac{1}{\gamma})$, comparing to the vanilla analysis that would have required $w$ to be at least $O(\frac{1}{\epsilon} \log k)$.

\begin{lemma} \label{lem:pabt-utility}
Let $\gamma, w > 0$ be any real numbers. Define $I_{\good} := \{i \in [k] \mid f_i(X) \geq T + w\}$ and $I_{\bad} := \{i \in [k] \mid f_i(X) \in (T - w, T + w)\}$. 
Suppose that the following conditions all hold:
\begin{itemize}
\item[(i)] $|I_{\good}| \geq \gamma \cdot k$.
\item[(ii)] $w \geq 8 \cdot \frac{1}{\epsilon} \cdot \log \frac{400}{\gamma}$.
\item[(iii)] $|I_{\good}| \geq 2 \cdot |I_{\bad}|$.
\end{itemize}
Then,
$\Pr_{i^*}[i^* \in I_\good] \geq 0.6$.
\end{lemma}

\begin{proof}
For a permutation $\pi$, let $j_{\good} \in [k]$ denote the smallest index such that $\pi(j_{\good}) \in I_{\good}$. Similarly, let $j_{\bad} \in [k]$ denote the smallest index such that $\pi(j_{\bad}) \in I_{\bad}$. Furthermore, let us define the following three events:
\begin{itemize}
\item Event $E_{\text{before}}$: $j_{\good} < j_{\bad}$.
\item Event $E_{\text{small-index}}$: $j_{\good} \leq 5/\gamma$.
\item Event $E_{\text{small-noise}}$: $|\rho|, |\nu_1|, \dots, |\nu_{\lceil 5/\gamma \rceil}| < w/2$ (where $\rho, \nu_1, \dots, \nu_{\lceil 5/\gamma \rceil}$ are the random variables sampled in the call to Algorithm~\ref{alg:abt}).
\end{itemize}
It is simple to see that, when the three events occur together, we have that $i^* = \pi(j_{\good}) \in I_{\good}$ as desired. Furthermore, we may bound the probability of each event as follows:
\begin{itemize}
\item Event $E_{\text{before}}$: this happens with probability exactly $\frac{|I_{\good}|}{|I_{\good}| + |I_{\bad}|}$, which is at least 2/3 from condition (iii).
\item Event $E_{\text{small-index}}$: the probability that this event does \emph{not} occur is at most
\begin{align*}
\left(1 - \frac{|I_{\good}|}{k}\right)^{\lfloor 5/\gamma\rfloor} \leq \left(1 - \gamma\right)^{4/\gamma} \leq e^{-4} \leq 0.02,
\end{align*}
where the first inequality follows from condition (i).
\item Event $E_{\text{small-noise}}$: Notice that each of $|\rho|, |\nu_1|, \dots, |\nu_{\lceil 5/\gamma \rceil}|$ is at least $w/2$ with probability at most
\begin{align*}
2 \exp\left(-\frac{w/2}{4/\eps}\right) \leq 2\left(\frac{\gamma}{400}\right) =
\frac{\gamma}{200},
\end{align*}
where the first inequality follows from   condition (ii). From a union bound, we have that $E_{\text{small-noise}}$ holds with probability at least 0.95.
\end{itemize}
Applying a union bound over all the three events, we can conclude that they all simultaneously hold with probability at least 0.6. This concludes our proof.
\end{proof}

\section{Iterative Sparse Vector Algorithm}
\label{sec:isvc}

Our iterative version of the correction algorithm via the sparse vector technique is presented in Algorithm~\ref{alg:svt-iterative}. As stated earlier, the algorithm performs the correction in multiple stages. In stage $\ell$, we use $\eps_\ell$ to denote the privacy parameter for each correction, $m_\ell$ to denote the number of corrections made, and $T_\ell$ to denote the threshold used. These parameters will be set below. Before doing so, let us state the guarantee that this algorithm achieves:

\begin{theorem} \label{thm:iterative}
For any $k \in \N$, $\eps \in (0, 1]$ and $\delta \in (0, 0.5]$, there exists an $(\eps, \delta)$-DP algorithm that given queries $q_1, \dots, q_k$ each of sensitivity at most one, outputs $a_1, \dots, a_k$ that satisfy
\begin{align*}
|\{i \in [k] \mid |q_i - a_i| > O(\sbound)\}| \leq O(k / (\log k)^{10})
\end{align*}
with probability $2^{-\Omega(k / (\log k)^{10})}$.
\end{theorem}

Notice that this guarantee does not give the $\ell_\infty$ error bound yet, as there can still be as many as $O(k/(\log k)^{10})$ coordinates that have error larger than the desired bound of $O(\sbound)$. However, as we will see in the next section, this already suffices for us to apply Theorem~\ref{thm:sv-correction} at the end and get a high probability bound on the $\ell_\infty$ error.

We remark that it is possible to select parameters in such a way that the final application of Theorem~\ref{thm:sv-correction} is not needed, i.e., by adding one more stage that essentially imitates Theorem~\ref{thm:sv-correction}. Nonetheless, since this does not seem to help clarify the analysis, we choose to not include it.

\begin{algorithm}[!htp]
\begin{algorithmic}[1]
  \Procedure{$\svt^{\eps_1, \dots, \eps_L}_{m_1, \dots, m_L, T_1, \dots, T_L}$}{$q_1, \dots, q_k$}
  \State $(a_1, \dots, a_k) \leftarrow (\infty, \dots, \infty)$
  \For{$\ell = 1, \dots, L$}
  \For{$j = 1, \dots, m_\ell$}
  \State $i^* \leftarrow \pabt^{0.5\eps_\ell}_{T_\ell}(|q_1 - a_1|, \dots, |q_k - a_k|)$ \label{step:select-coordinate}
  \State $a_{i^*} \leftarrow q_{i^*}(X) + \Lap(2/\eps_\ell)$ \label{step:resampled}
  \EndFor
  \EndFor
  \State \Return $(a_1, \dots, a_k)$
  \EndProcedure
 \end{algorithmic}
\caption{Iterative Sparse Vector Correction Algorithm.}\label{alg:svt-iterative}
\end{algorithm}

Our selection of parameters is as follows:
\begin{itemize}
\item $\kappa = 0.9$ and $\lambda = 0.95$
\item $L = \lceil10 \log_{1/\kappa} \log k\rceil$
\item $\eps_0 = \frac{\eps}{1000\sqrt{\log(1/\delta)}}$
\end{itemize}
For $\ell \geq 1$, we define
\begin{itemize}
\item $m_{\ell} = \kappa^{\ell} \cdot k$
\item $\eps_\ell = \frac{\eps_0}{\sqrt{k}} \left(\frac{1}{\sqrt{\ell \lambda^{\ell}}}\right)$ 
\item $w_\ell = \frac{100 \log(500 k / m_\ell)}{\eps_\ell}$
\item $T_\ell =  4(w_1 + \dots + w_{\ell-1}) + 3 w_{\ell} + 2 w_{\ell+1}$
\end{itemize}
Finally, we define $T_0 = 2w_1$ and $w_0 = 0$.

Throughout this section, we always assume that the parameters are as specified above and we will not mention this again. We also assume that $m_\ell$ defined above is an integer for every $\ell \in [L]$. This is without loss of generality since we may simply replace $k$ with $k' := 10^{\lceil \log_{10} k\rceil}$ where $q_{k + 1}, \dots, q_{k'}$ are constants; when $k$ is sufficiently large, this ensures that $m_\ell$ is an integer for all $\ell \in [L]$.

The proof of Theorem~\ref{thm:iterative} is broken down into two parts: the privacy proof and the utility proof. 

\subsection{Privacy Analysis}

We will start by proving the privacy guarantee of the algorithm.

\begin{theorem}[Privacy Guarantee]
For $\eps \in (0, 1]$ and $\delta \in (0, 0.5]$, Algorithm~\ref{alg:svt-iterative} is $(\eps, \delta)$-DP. 
\end{theorem}

\begin{proof}
From Observation~\ref{obs:pabt-privacy} and the privacy of the Laplace mechanism, we can conclude that a single execution of Lines~\ref{step:select-coordinate} and~\ref{step:resampled} is $\eps_\ell$-DP. Hence, for a fixed outer iteration $\ell \in [L]$, we may apply advanced composition (Theorem~\ref{thm:adv-composition}) to conclude that it is $(\eps'_\ell, \delta'_\ell)$-DP where $\delta'_\ell = 0.5^{\ell} \cdot \delta$ and
\begin{align*}
\eps'_\ell &= \sqrt{2m_\ell \log(2^\ell / \delta)} \cdot \eps_\ell + m_\ell \eps_{\ell}(e^{\eps_\ell} - 1) \\
&\leq \sqrt{2m_\ell \log(2^\ell / \delta)} \cdot \eps_\ell + 2 m_\ell \eps_{\ell}^2 \\
&= \sqrt{2m_\ell \log(2^\ell / \delta)} \cdot \frac{\eps_0}{\sqrt{k}} \left(\frac{1}{\sqrt{\ell \cdot \lambda^{\ell}}}\right) + 2m_\ell \left(\frac{\eps_0}{\sqrt{k}} \left(\frac{1}{\sqrt{\ell \cdot \lambda^{\ell}}}\right)\right)^2 \\
&\leq \eps_0\left(\sqrt{\frac{2m_\ell}{\lambda^{\ell} k} \cdot \frac{\log(2^\ell / \delta)}{\ell}} + \frac{2m_\ell}{\lambda^{\ell} k} \right) \\
&\leq 4 \eps_0 \left(\sqrt{(\kappa/\lambda)^{\ell} \log(1/\delta)}\right) \\
&\leq \frac{\eps}{200} \cdot (\kappa/\lambda)^{\ell/2}.
\end{align*}

Finally, we apply basic composition (Theorem~\ref{thm:basic-composition}) over all $\ell \in [L]$, which implies that the entire algorithm is $(\eps', \delta')$-DP for
\begin{align*}
\eps' = \sum_{\ell \in [L]} \eps'_\ell = \sum_{\ell \in [L]} \frac{\eps}{200} \cdot (\kappa/\lambda)^{\ell/2} \leq \frac{\eps}{200(1 - (\kappa/\lambda)^{0.5})} \leq \eps,
\end{align*}
and
\begin{align*}
\delta' = \sum_{\ell \in [L]} \delta'_\ell = \sum_{\ell \in [L]} \left(0.5^{\ell} \delta\right) \leq \delta,
\end{align*}
as desired.
\end{proof}

\subsection{Utility Analysis}

We will next prove the utility guarantee, as restated below.  

\begin{theorem}[Utility Guarantee] \label{thm:utility}
With probability at least $1-2^{-\Omega(k / (\log k)^{10})}$, the output $(a_1, \dots, a_k)$ of Algorithm~\ref{alg:svt-iterative} satisfies 
\begin{align*}
|\{i \in [k] \mid |q_i - a_i| > O(\sbound) \}| \leq O(k / (\log k)^{10})
\end{align*}
\end{theorem}

Our utility analysis crucially relies on tracking the set of indices $i$ such that $|q_i - a_i|$ is above a certain threshold.
Specifically, for every $\ell \in [L], j \in [m_\ell]$ and $t \in \{0, \dots, L\}$, we define $I_{t}^{\ell, j}$ to be the set of indices $i = 1, \dots, k$ such that, after the $(\ell, j)$-th iteration, $|q_i - a_i| \geq T_t + w_t$. For notational convenience, we define $I_t^\ell := I_t^{\ell, m_\ell}$, $I_t^0 := [k]$, $I^{\ell, 0}_t := I^{\ell - 1}_t$, and
$\tau_\ell := T_\ell + w_\ell$.  

The high-level idea of the proof is to consider two cases, based  on whether the number of indices at the end of the $(\ell - 1)$-th iteration whose errors exceed $\tau_\ell$, i.e., $|I^{\ell - 1}_{\ell}|$, is small. Now, if this is already small (i.e., noticeably smaller than $m_\ell$), then we can use a concentration inequality to show that the number of additional indices that are ``flipped'' from below $\tau_\ell$ to above $\tau_\ell$ is small; from this, we can conclude that $|I^\ell_\ell|$ is small. On the other hand, if $|I^{\ell - 1}_{\ell}|$ is large, then we know that after the $(\ell -  1)$-th iteration the number of indices whose errors belong to $(\tau_{\ell - 1}, \tau_\ell)$ is small.  Roughly speaking, this allows us to apply Lemma~\ref{lem:pabt-utility}, which ensures that a significant fraction of selected coordinates indeed have errors at least $\tau_\ell$. This means that $|I^{\ell}_{\ell}|$ must be significantly smaller than $|I^{\ell - 1}_{\ell}|$, which ultimately gives us the desired bound in the second case.

We will show that with high probability $|I_{\ell}^\ell|$ is small, as formalized in our main lemma below.

\begin{lemma} \label{lem:one-stage-reduction}
Let $\ell \in [L]$.
Conditioned on $|I^{\ell - 1}_{\ell - 1}| \leq 2m_{\ell - 1}$, we have that
\begin{align*}
\Pr[|I^\ell_\ell| \leq 2m_\ell] \geq 1 - 2^{-\Omega(m_\ell)}. 
\end{align*}
\end{lemma}

\begin{proof}
Consider the $\ell$-th (outer) iteration of the algorithm.
Let $Z_j$ denote the Laplace random variable drawn on Line~\ref{step:resampled} in the $j$-th inner iteration. Notice that from our setting of parameters $\eps_\ell, \tau_{\ell - 1}$, and $w_\ell$, we have
\begin{align*}
\Pr[|Z_j| \geq \tau_{\ell - 1}] \leq \Pr[|Z_j| \geq w_\ell] \leq 0.0009.
\end{align*}
From this and from the independence of the $Z_j$'s, we may apply the Chernoff bound, which implies that the following holds with probability at least $1 - 2^{-\Omega(m_\ell)}$:
\begin{align} \label{eq:resampling-error}
|\{j \in [m_\ell] \mid Z_j \geq \tau_{\ell - 1}\}| \leq 0.001m_\ell.
\end{align}
Thus, we may hence forth assume that~\eqref{eq:resampling-error} holds.

We will next consider two cases:
\begin{enumerate}
\item Case I: $|I^{\ell-1}_\ell| \leq 1.999 m_\ell$. From~\eqref{eq:resampling-error}, it follows that $|I^\ell_\ell| \leq |I^{\ell-1}_\ell| + 0.001m_\ell \leq 2m_\ell$ as desired.
\item Case II: $|I^{\ell-1}_\ell| > 1.999 m_\ell$. 

In this case, we would like to apply Lemma~\ref{lem:pabt-utility}. To do this, we will check that each of condition of Lemma~\ref{lem:pabt-utility} is satisfied with $T = T_\ell, w = w_\ell, \gamma = 0.9m_\ell/k$, and $\eps = \eps_\ell$.

We can bound the number of ``good'' indices $i$ such that $|q_i - a_i| \geq \tau_\ell$ by
\begin{align} \label{eq:lb-above-thresholds}
|I^{\ell, j}_\ell| \geq |I^{\ell -  1}_\ell| - j \geq 1.999m_\ell - m_\ell \geq 0.999m_\ell,
\end{align}
which is at least $\gamma \cdot k$ as desired.

The second condition of Lemma~\ref{lem:pabt-utility} holds simply by our choice of $w_\ell$.

Finally, let us bound the number of indices $i$ such that $a_i \in [\tau_{\ell - 1}, \tau_\ell)$ as follows.  Notice that
\begin{align*}
|I^{\ell - 1}_{\ell - 1} \setminus I^{\ell - 1}_{\ell}| \leq m_{\ell - 1} - 0.999m_{\ell} = (1/\kappa - 1.999)m_\ell \leq 0.23m_\ell,
\end{align*}
where the first inequality follows from our assumption on $I^{\ell - 1}_{\ell - 1}$.
We may now use~\eqref{eq:resampling-error} to conclude that, for any $j \in [m_\ell]$, we have
\begin{align} \label{eq:ub-between-thresholds}
|I_{\ell - 1}^{\ell, j} \setminus I_{\ell}^{\ell, j}| \leq 0.23m_\ell + 0.001m_\ell \leq 0.3m_\ell.
\end{align}
In other words, in the $\ell$-th outer loop, the number of indices $i$ with $a_i \in [\tau_{\ell - 1}, \tau_\ell)$ is always at most $0.3m_\ell$. Since $T_{\ell} = \tau_\ell - w_\ell \geq \tau_{\ell - 1} + w_\ell$, this also implies that the number of ``bad'' indices $i$ with $a_i \in [T_{\ell} - w_\ell, T_{\ell} + w_\ell)$ is at most $0.3m_\ell$. Together with~\eqref{eq:lb-above-thresholds}, this implies that the last condition of Lemma~\ref{lem:pabt-utility} holds.

Thus, we may apply Lemma~\ref{lem:pabt-utility}, which implies that $\Pr[a_{i^*} \geq \tau_\ell] \geq 0.5$ for each call on Line~\ref{step:select-coordinate}. Hence, by the Chernoff bound, we can conclude that, with probability $2^{-\Omega(m_\ell)}$, at least $0.4m_\ell$ of the $i^*$'s returned satisfy $a_{i^*} \geq \tau_\ell$. When this holds, it (together with~\eqref{eq:resampling-error}) implies that 
\begin{align*}
|I^\ell_\ell| 
&\leq |I^{\ell - 1}_{\ell}| - 0.4 m_{\ell} + 0.001m_\ell \\
&\leq |I^{\ell - 1}_{\ell - 1}| - 0.4m_\ell + 0.001m_\ell \\
&\leq 2m_{\ell - 1} - 0.4m_\ell + 0.001m_\ell \\
&\leq 2m_\ell.
\end{align*}
\end{enumerate}

Hence, in both cases, we have that with probability $2^{-\Omega(m_\ell)}$, it holds that $|I^\ell_\ell| \leq 2m_\ell$ as desired.
\end{proof}

Theorem~\ref{thm:utility} now follows easily from the above lemma.

\begin{proof}[Proof of Theorem~\ref{thm:utility}]
By applying Lemma~\ref{lem:one-stage-reduction} for each $\ell \in [L]$ and a union bound, we have that
\begin{align*}
\Pr[|I^L_L| \geq 2m_L] \leq 1 - \sum_{\ell \in [L]} 2^{-\Omega(m_\ell)} \leq 1 - 2^{-\Omega(k/(\log k)^{10})},
\end{align*}
where the latter follows from $m_1 \geq \cdots \geq m_L = \Theta(k/(\log k)^{10})$. Finally, notice that
\begin{align*}
\tau_L &\leq 4\left(\sum_{\ell \in [L + 1]} w_\ell\right) \\
&= 4\left(\sum_{\ell \in [L + 1]} \frac{100 \log(500 k / m_\ell)}{\eps_\ell}\right) \\
&= 4\frac{\sqrt{k}}{\eps_0} \left(\sum_{\ell \in [L + 1]} 100 \log(500 (1/\kappa)^\ell) \cdot \sqrt{\ell \lambda^\ell}\right) \\
&\leq O(\sqrt{k}/\eps_0) \cdot \left(\sum_{\ell \in [L + 1]} \ell^{3/2} \cdot \lambda^{\ell/2}\right) \\
&= O(\sqrt{k}/\eps_0) \\
&= O(\sbound).
\end{align*}
This means that $I^L_L = \{i \in [k] \mid |q_i - a_i| > O(\sbound)\}$ as desired.
\end{proof}

\section{Obtaining High Probability Error Bound}
\label{sec:high-prob-error}

In this section, we use our bound in the previous section together with Theorem~\ref{thm:sv-correction} to obtain the following high probability $\ell_\infty$ error guarantee:

\begin{theorem} \label{thm:high-prob-error}
For any $k \in \N, \eps \in (0, 1]$ and $\delta \in (0, 0.5]$, there exists an $(\eps, \delta)$-DP algorithm that can provide answers to $k$ queries such that, with probability $1 - O(1/k^{10})$, the $\ell_\infty$ error is $O(\sbound)$.
\end{theorem}


\begin{proof}
First, we apply the $(\eps/2, \delta/2)$-DP algorithm from Theorem~\ref{thm:iterative} to get the answers $a_1, \dots, a_k$ to the queries. From Theorem~\ref{thm:iterative}, there exist constants $C_1, C_2 > 0$ such that, with probability $2^{-\Omega(k / (\log k)^{10})}$, we have that
\begin{align*}
|\{i \in [k] \mid |q_i - a_i| > C_1 \cdot \sbound \}| \leq C_2 \cdot k / (\log k)^{10}.
\end{align*}
We then apply Theorem~\ref{thm:sv-correction} with $\eps_{\SV} = \eps/2, \delta_{\SV} = \delta/2, \beta_{\SV} = 1 / k^{10}, c_{\SV} = C_2 \cdot k / (\log k)^{10}, g_i = q_i - a_i$, and let $\alpha_{\SV} = 2C_1 \cdot \sbound$. Let $b_1, \dots, b_k$ be its output. Our algorithm then outputs $a_1 + b_1, \dots, a_k + b_k$.

It is simple to verify that the condition on $\alpha_{\SV}$ in Theorem~\ref{thm:sv-correction} holds for any sufficiently large $k$. As a result, Theorem~\ref{thm:sv-correction} implies that, with probability $1 - O(1/k^{10})$, the $\ell_\infty$ error is at most $O(\sbound)$.
\end{proof}

\section{From High Probability to Expected Error Bound}
\label{sec:expected-error}

Finally, we will transform our high probability error bound into an expected error bound (Theorem~\ref{thm:main}). To do so, we need the following well-known theorem, which follows from the Gaussian mechanism (see e.g.~\cite{DworkR14,SteinkeU16}).

\begin{theorem}[Gaussian Mechanism] \label{thm:gaussian}
For any $k \in \N, \eps_g \in (0, 1]$ and $\delta_g \in (0, 0.5]$, there exists an $(\eps_g, \delta_g)$-DP algorithm that can answer $k$ queries such that, for any $t > 0$, the $\ell_\infty$ error is at least $O(t \cdot \bound{k}{\eps_g}{\delta_g})$ with probability at most $k \cdot e^{-\Omega(t^2)}$.
\end{theorem}

\begin{proof}[Proof of Theorem~\ref{thm:main}]
The entire algorithm is as follows:
\begin{itemize}
\item First, we apply $(\eps/3, \delta/3)$-DP algorithm from Theorem~\ref{thm:high-prob-error} to get the answers $a_1, \dots, a_k$ to the input queries.
\item Secondly, we apply $(\eps/3, \delta/3)$-DP algorithm from Theorem~\ref{thm:high-prob-error} to get the answers $b_1, \dots, b_k$ to the input queries.
\item Then, we apply $(\eps/3, \delta/3)$-DP algorithm from Theorem~\ref{thm:gaussian} to the queries $|q_1 - a_1|, \dots, |q_k - a_k|$ to get the answers $c_1, \dots, c_k$
\item If $\max\{c_1, \dots, c_k\} \leq k^{10} \cdot \sbound$, then we output $(a_1, \dots, a_k)$.
\item Otherwise, if $\max\{c_1, \dots, c_k\} > k^{10} \cdot \sbound$, then we output $(b_1, \dots, b_k)$.
\end{itemize}
By basic composition (Theorem~\ref{thm:basic-composition}), the entire algorithm is $(\eps, \delta)$-DP as desired.

We next analyze the expected $\ell_\infty$ error of the algorithm. 
To do this, notice first that, regardless of $a_1, \dots, a_k$, the $\ell_\infty$ error of the output is at most the sum of $k^{10} \cdot \sbound$ and the two $\ell_\infty$ errors of the two runs of the Gaussian mechanism. As such, we still have that the $\ell_\infty$ error of the entire algorithm is at least $k^{10} \cdot \sbound + O(t \cdot \sbound)$ with probability at most $k \cdot e^{-\Omega(t^2)}$.

Furthermore, the guarantee from Theorem~\ref{thm:high-prob-error} ensures that the $\ell_\infty$ error of the answers $a_1, \dots, a_k$ is at most $O(\sbound)$ with probability $1 - O(1/k^{10})$. When this event holds, we may apply the tail bound from Theorem~\ref{thm:gaussian}, which implies that the output of the entire algorithm has error $h := O(\sbound)$ with probability $1 - O(1/k^{10})$.

Let $\upsilon$ denote the $\ell_\infty$ error of the entire algorithm. Combining the bounds from the previous two paragaphs, we get
\begin{align*}
\E[\upsilon] &= \int_0^\infty \Pr[\upsilon > x] dx \\
&= \int_0^h \Pr[\upsilon > x] dx + \int_h^{2k^{10} \cdot \sbound} \Pr[\upsilon > x] dx + \int_{2k^{10} \cdot \sbound}^{\infty} \Pr[\upsilon > x] dx \\
&\leq h + O(1/k^{10}) \cdot (2k^{10} \cdot \sbound) + O(\sbound) \cdot \int_{k^{10}}^{\infty} e^{-\Omega(t^2)} dt \\
&\leq O(\sbound),
\end{align*}
which concludes our proof.
\end{proof}

\section{Conclusions and Open Questions}
\label{sec:open}

In this work, we give an $(\eps, \delta)$-DP algorithm that can answer $k$ queries, each of sensitivity one, with $\ell_{\infty}$ error $O(\sbound)$. This resolves the question posed by Steinke and Ullman~\cite{SteinkeU16}.

An immediate open question is if one can get the best of both our work and~\cite{DK20}. Namely, to devise an $(\eps, \delta)$-DP algorithm for answering $k$ queries (of sensitivity one) such that the error is always (i.e., with probability one) $O(\sbound)$ for any value of $\delta > 0$.


\section*{Acknowledgments}

We are grateful to Thomas Steinke for introducing us to the problem and for explaining to us useful insights from previous works, especially from~\cite{SteinkeU16}.

\bibliographystyle{alpha}
\bibliography{refs}

\end{document}